\documentclass[11pt]{article}

\usepackage{amssymb,amsmath,enumerate,graphicx,version,amsfonts,amsthm}
\usepackage[utf8]{inputenc}
\usepackage{color}

\usepackage[english]{babel}

\newtheorem{Prop}{PROPOSITION}

\newtheorem{Rem}{REMARK}

\newcommand{\beq}{\begin{equation}}
\newcommand{\eeq}{\end{equation}}
\newcommand{\pa}{\partial}

 \DeclareMathOperator{\im}{Im}
 \DeclareMathOperator{\hess}{hess}

\DeclareMathOperator{\pf}{pf} \DeclareMathOperator{\tr}{tr}

\DeclareMathOperator{\vol}{vol}

\begin{document}

\title{\textsf{Monge-Amp\`ere Structures and the \\ Geometry of Incompressible Flows}}
\author{B. Banos\thanks{LMBA, Universit\'e de Bretagne Sud, Campus de Tohannic  
BP 573, 56017 Vannes, France},  
V. N. Roubtsov\thanks{UNAM, LAREMA, D\'epartment de Math\'ematiques, Universit\'e d'Angers, 2 Blvd Lavoisier, 
49045 Angers, France; \& Theory Division, ITEP, 25, 
Bol. Tcheremushkinskaya, Moscow 117259, Russia}, \& I. Roulstone\thanks{Department of Mathematics, University of Surrey, 
Guildford GU2 7XH, UK. {\em email}: i.roulstone@surrey.ac.uk}}

\date{\today}

\maketitle

\begin{abstract}
We show how a symmetry reduction of the equations for incompressible hydrodynamics in three dimensions leads naturally to Monge-Amp\`ere structure, and Burgers'-type vortices are a canonical class of solutions associated with this structure. The mapping of such solutions, which are characterised by a linear dependence of the third component of the velocity on the coordinate defining the axis of rotation, to solutions of the incompressible equations in two dimensions is also shown to be an example of a symmetry reduction. The Monge-Amp\`ere structure for incompressible flow in two dimensions is shown to be hyper-symplectic.
\end{abstract}

\section{Introduction}

The incompressible Euler and Navier--Stokes equations pose an enduring challenge for mathematical analysis. Previous works by Roubtsov and  Roulstone (1997, 2001), Delahaies and Roulstone (2010), and by McIntyre and Roulstone (2002), have shown how quite sophisticated  ideas from modern differential  geometry offer a framework  for understanding how coherent structures in large-scale atmosphere-ocean flows, such as cyclones and fronts, may be modelled using a hierarchy of approximations to the Navier--Stokes equations. These studies focussed  on semi-geostrophic (SG)  theory and  related models. A Monge-Amp\`ere (MA) equation lies at the heart of these models, and it was observed that geometry is the natural language for describing the interconnections  between  the  transformation theory of the MA equation (Kushner {\em et al.} 2007), the  Hamiltonian properties of fluid dynamical  models, and particular features of solutions such as stability and discontinuity. 

In this  paper, we extend the more recent work of  Roulstone {\em et al.} (2009a,b) on the Euler  and Navier--Stokes  equations, and  show how the geometric approach  to SG  theory may be applied to the incompressible fluid equations in three dimensions: in particular, we show how vortex flows of Burgers' type are described by a Monge-Amp\`ere structure. We also show how the partial differential equations that describe such canonical flows are obtained via a symmetry reduction.

This paper is organised as follows. We start by considering Monge-Amp\`ere structures in four dimensions and incompressible flow in two dimensions. 
In Section~3, we show how a hyper-symplectic structure is associated with these equations. 
In Section~4, we review Monge-Amp\`ere structures in six dimensions and the reduction principle. 
In Section~5 we show how Burger's vortices in three-dimensional flow relate to Monge-Amp\`ere structures in six dimensions. 
We then explain how a symmetry reduction maps this to a Monge-Amp\`ere, hyper-symplectic, structure in four dimensions. 
This reduction is surely known to fluid dynamicists in the context of the Lundgren transformation (Lundgren 1982). 
In Section~6 we show how a symmetry reduction of the incompressible Euler equations in three dimensions yields the generic geometric 
framework exploited in Section~5. A summary is given in Section~7.

\section{The geometry of Monge-Amp\`ere equations}

\subsection{Monge-Amp\`ere structures}

Lychagin (1979) proposed a geometric approach to to study of MA equations using differential forms. The idea is to associate with a form $\omega\in
\Lambda^n(T^*\mathbb{R}^n)$, where $\Lambda^n$ denotes the space of differential $n$-forms 
on $T^*\mathbb{R}^n$, the MA equation $\Delta_\omega=0$, where  
$\Delta_{\omega} : C^\infty(\mathbb{R}^n)\rightarrow
\Omega^n(\mathbb{R}^n)\cong C^\infty(\mathbb{R}^n)$ is the
differential operator defined by
$$
\Delta_{\omega}(\phi)=(d\phi)^*\omega\,,
$$
and $(d\phi)^*\omega$ denotes the restriction of $\omega$ to the graph of $d\phi$ 
($d\phi: \mathbb{R}^n\rightarrow T^*\mathbb{R}^n$ is the differential of $\phi$).  

Following Lychagin {\em et al.} (1993), a MA equation in $n$
variables can be described by a pair $(\Omega,\omega)\in
\Lambda^2(T^*\mathbb{R}^n)\times \Lambda^n(T^*\mathbb{R}^n)$ such
that
\begin{enumerate}[i)]
\item $\Omega$ is symplectic that is nondegenerate and closed
\item $\omega$ is effective, that is $\omega\wedge\Omega=0$.
\end{enumerate}
Such a pair is called a Monge-Amp\`ere structure.

A  generalized  solution of a MA equation is a lagrangian submanifold $L$, on which $\omega$ vanishes; that is, $L$ is an $n$-dimensional 
submanifold such that $\Omega|_L=0$ and $\omega|_L=0$. Note  that  if $\phi(x)$  is  a
regular solution then  its graph $L_\phi=\{(x,\phi_x)\}$ in  the phase space
is a generalized solution. 


In four dimensions (that is $n=2$), a geometry defined by this structure can be either
complex or real and this distinction coincides with the
usual distinction between elliptic and hyperbolic differential equations
in two variables. Indeed,  when $\omega\in
\Lambda^2(T^*\mathbb{R}^2)$ is a nondegenerate $2$-form
($\omega\wedge\omega\neq 0$), one can associate with  the
Monge-Amp\`ere structure $(\Omega,\omega)\in
\Lambda^2(T^*\mathbb{R}^2)\times \Lambda^2(T^*\mathbb{R}^2)$ the
tensor $I_\omega$ defined by
\beq
\label{eq:Iw}
\frac{1}{\sqrt{|\pf(\omega)|}}\omega(\cdot,\cdot)=\Omega(I_\omega\cdot,\cdot)
\eeq
where $\pf(\omega)$ is the pfaffian of $\omega$:
$\omega\wedge\omega = \pf(\omega)(\Omega\wedge\Omega)$. This
tensor is either an almost complex structure or an almost product
structure:
\begin{enumerate}[a)]
\item $\Delta_\omega=$ is elliptic $\Leftrightarrow$
$\pf(\omega)>0$ $\Leftrightarrow$ $I_\omega^2=-Id$
\item $\Delta_\omega=$ is hyperbolic $\Leftrightarrow$
$\pf(\omega)<0$ $\Leftrightarrow$ $I_\omega^2=Id$
\end{enumerate}
Lychagin and Rubtsov (1983a,b, 1993) (see also Kushner {\em et al.} (2007) for a comprehensive account of this theory) have explained the
link there is between the problem of local equivalence of MA equations in two
variables and the integrability problem of this tensor $I_\omega$:
\begin{Prop} (Lychagin-Roubtsov theorem) 
The three following assertions are
equivalent:
\begin{enumerate}
\item $\Delta_\omega=0$ is locally equivalent to one of the two
equations
$$
\begin{cases}
\Delta\phi=0&\\
\square\phi=0&\\\end{cases}
$$
\item the almost complex (or product) structure $I_\omega$ is
integrable
\item the form $\frac{\omega}{\sqrt{|\pf(\omega)|}}$ is closed.
\end{enumerate}
\end{Prop}

\section{The  geometry of  the incompressible  Euler equations  in two
  dimensions} 

\subsection{The Poisson equation for the pressure}

The incompressible Euler equations (in two or three dimensions) are
\beq
\label{e1}
\frac{\pa u_i}{\pa t} + u_j u_{ij} + p_i = 0, \;\;\;\;
u_{ii}=0 ,
\eeq
with $ p_i = \pa p/\pa x_i, \;\;u_{ij} = \partial u_i/\partial x_j$, and the summation convention is used.
By applying the incompressibility condition, we find that the Laplacian of the pressure
and the velocity gradients are related by 
\begin{equation}
\label{2DP}
\Delta p = -u_{ij}u_{ji} .
\end{equation}

In two dimensions, the incompressibility condition implies 
$$
u_1 = -\frac{\partial \psi}{\partial x_2} = -\xi_2,\;\;
u_2 = \frac{\partial \psi}{\partial x_1} = \xi_1 ,
$$
where  $\psi(x_1, x_2,  t)$  is the  stream  function and  \eqref{2DP} becomes 
\begin{equation}
\psi_{x_1x_1}\psi_{x_2x_2}-\psi_{x_1x_2}^2 = \frac{\Delta p}{2} .
\label{2DE}
\end{equation}
Equation (\ref{2DP})  is normally used  as a Poisson equation  for the
pressure. However, formally,  we may consider (\ref{2DE}) to be  an equation of
Monge-Amp\`ere type for  the stream function, if the  Laplacian of the
pressure  is  given.  Roulstone  {\em   et  al.}  (2009a)  discuss  the
relationship between  the sign of  the Laplacian of the  pressure, the
elliptic/hyperbolic  type of  the associated  MA equation,
and the balance between rate of strain and the enstrophy of the flow. This is related to the Weiss criterion (Weiss 1991), as discussed by Larcheveque (1993).

\subsection{Hypersymplectic geometry}

Following Section 2.1, $\eqref{2DE}$ can be associated with a Monge-Amp\`ere structure $(\Omega,\omega)$ on the
phase    space     $M=T^*\mathbb{R}^2$,    where    $\Omega=dx_1\wedge
du_2+du_1\wedge dx_2$ 
is a symplectic form (which is written conventionally in the form $\Omega = dx_1\wedge d\xi_1 + dx_2\wedge d\xi_2$) and $\omega\in \Omega^2(M)$  is the
$2$-form defined by 
$$
\omega = du_1\wedge du_2 - a(x_1, x_2) dx_1\wedge dx_2
$$
with 
$$
a(x_1,x_2, t)=\frac{\Delta p}{2} .
$$
(Note that time, $t$, is a parameter in our examples, and its presence
will not be indicated, unless crucial.) 

We observe that
\begin{equation}
d\omega=0
\label{do}
\end{equation}
and
\begin{equation}
\omega\wedge \omega = a \,\Omega\wedge \Omega .
\label{pf}
\end{equation}
Moreover,  the  tensor $A_\omega$  (cf. (\ref{eq:Iw})) defined  by $\omega(\cdot,\cdot)  =
\Omega(A_\omega\cdot,\cdot)$ satisfies 
$$
A_\omega^2 = -a .
$$

Following Roubtsov and Roulstone (2001) (see also Kossowski (1992)), we define a symmetric tensor $g_\omega$ on $M$
$$
g_\omega(X,Y)=\frac{2\left(\iota_X\omega\wedge
    \iota_Y\Omega+\iota_Y\omega\wedge   \iota_X\Omega\right)\wedge  dx_1
  \wedge dx_2}{\Omega^2} .
$$
In coordinates, one has 
\begin{equation}
\label{g}
g_\omega = dx_1\otimes du_2  - dx_2\otimes du_1 
\end{equation}
and therefore the signature of $g_\omega$ is $(2,2)$.
Using this split metric we define a dual form $\hat{\omega}\in
\Lambda^2(M)$ 
$$
\hat{\omega}(\cdot,\cdot) = g_\omega(A_\omega\cdot,\cdot) ,
$$
which in coordinates is
\begin{equation}
\label{ho}
\hat{\omega} =-du_1\wedge du_2-a(x_1,x_2)dx_1\wedge dx_2 .
\end{equation}
This dual form satisfies 
\begin{equation}
d\hat{\omega}=0
\label{dho}
\end{equation}
and
\begin{equation}
\hat{\omega}\wedge \hat{\omega} = - a \,\Omega\wedge \Omega .
\label{pfho}
\end{equation}

Hence, introducing  the normalized  form $\tilde{\Omega}  = \sqrt{|a|}
\Omega$, we obtain a triple $\left(\tilde{\Omega},
  \omega,\hat{\omega}\right)$, which is  hypersymplectic (Hitchin 1990) when $\Delta p
\neq 0$: 
$$
\begin{array}{ccc}
\omega^2 = -\hat{\omega}^2, & \omega^2 = \varepsilon
\tilde{\Omega}^2, & \hat{\omega}^2 = -\varepsilon \tilde{\Omega}^2 ,\\
\omega  \wedge  \hat{\omega} =  0,  &  \omega\wedge \tilde{\Omega}=0,  &
\hat{\omega}\wedge \tilde{\Omega}=0 ,\\
\end{array}
$$
with $\varepsilon=1$ if $\Delta p>0$ and $\varepsilon=-1$ if $\Delta p<0$.

Equivalently, we obtain three tensors $S$, $I$ and $T$ defined by 
$$
\begin{array}{ccc}
\hat{\omega}(\cdot,\cdot) = \omega(S\cdot,\cdot), & \omega(\cdot,\cdot)
=    \tilde{\Omega}(I\cdot,\cdot),   &    \hat{\omega}(\cdot,\cdot)   =
\tilde{\Omega}(T\cdot,\cdot) ,\\
\end{array}
$$
which satisfy
$$
\begin{aligned}
&S^2 = 1, I^2 = -\varepsilon, T^2=\varepsilon ,\\
&TI=-IT=S,\\
&TS=-ST=I,\\
&IS=-SI=T.\\
\end{aligned}
$$
 This structure should be compared with the hyper-K\"ahler properties of the SG equations (Delahaies and Roulstone 2010), and it extends the results presented in Roulstone {\em et al.} (2009a).

Moreover, since 
$$
d\omega=d\hat{\omega} = 0\;\;\mbox{and}\;\; d\tilde{\Omega} = \frac{1}{2\sqrt{a}} \,da\wedge \Omega ,
$$
we deduce from Lychagin-Roubtsov theorem (see Proposition 1) that 
\begin{enumerate}[i)]
\item the product structure $S$ is always integrable
\item the complex structure $I$ and the product structure $T$ (or vice
  versa) are integrable if and only if $\Delta p$ is constant.
\end{enumerate}

\subsection{Generalized solutions}

A generalized solution of  \eqref{2DE} is a 2d-submanifold $L^2\subset
M^4$   which   is   bilagrangian   with  respect   to   $\omega$   and
$\Omega$. Since $\omega =  \Omega(I,\cdot,\cdot)$, it is equivalent to
saying that 
$L$ is  closed under $I$:  for any non  vanishing vector field  $X$ on
$L$, $\{X,IX\}$ is a local frame of $L$.  

\begin{Prop}  Let $L$  be a  generalized solution  and $h_\omega$  the
  restriction of $g_\omega$ on $L$.
\begin{enumerate}
\item if $\Delta p>0$, the  metric $h_\omega$ has signature $(2,0)$ or
  $(0,2)$ 
\item if $\Delta p<0$, the metric $h_\omega$ has signature $(1,1)$.
\end{enumerate}
\end{Prop}

\begin{proof}
Let  $X$ a  local non  vanishing vector  field, tangent  to $L$.  Then
$\{X,IX\}$ is a local frame of  $L^2$ and $\{X,IX,SX,TX\}$ a local frame
of $M^4$.
Let $\alpha =\hat{\omega}(X,IX)$. Then we have
\begin{enumerate}[i)]
\item $\omega(X,X)=0$ by skew-symmetry .
\item $\omega(X,IX)=0$ since $X$ and $IX$ are tangent to $L$.
\item $\omega(X,SX)=-\omega(SX,X)=\hat{\omega}(X,X)=0$.
\item $\omega(X,TX)=\omega(SIX,X)=\hat{\omega}(IX,X)=-\alpha$.
\end{enumerate}
Therefore, since $\omega$ is non degenerate, $\alpha$ cannot be $0$.
Moreover, 
\begin{enumerate}[i)]
\item $\displaystyle g_\omega(IX,IX) = \frac{1}{|a|}g_\omega(A_\omega
  X,A_\omega    X)=    \frac{1}{|a|}\hat{\omega}(X,A_\omega    X)    =
  \frac{\alpha}{\sqrt{|a|}}$ 
\item            $\displaystyle           g_\omega(IX,X)            =
  \frac{1}{\sqrt{|a|}}g_\omega(A_\omega                          X,X)=
  \frac{1}{\sqrt{|a|}}\hat{\omega}(X,X)=0$ 
\item $\displaystyle g_\omega(X,X)  = -\frac{1}{a} g_\omega(A^2_\omega
  X,   X)=   -\varepsilon\frac{1}{|a|}\hat{\omega}(A_\omega   X,X)   =
  \frac{\alpha \varepsilon}{\sqrt{|a|}}$
\end{enumerate}

We therefore obtain the matrix of $h_\omega$ in the representation $\{X,IX\}$:
$$
h_\omega = \frac{\alpha}{\sqrt{|a|}}\begin{pmatrix} \varepsilon  & 0 \\ 0
  &1\end{pmatrix} 
$$
\end{proof}

\begin{Rem}
If  $L=L_\psi$  is  a   regular  solution,  then  the  induced  metric
$h_\omega$ is affine. Indeed its tangent space is
generated  by $\displaystyle  X_1 =  \partial_{x_1} -  \psi_{x_1x_2} \partial_{u_1} +
\psi_{x_1x_1}\partial_{u_2}$ and $\displaystyle  X_2 =  \partial_{x_2} -  \psi_{x_2x_2} \partial_{u_1} +
\psi_{x_1x_2}\partial_{u_2}$. The induced metric on $L_\psi$ is therefore
\beq
\label{R1}
h_\omega=2\begin{pmatrix}\psi_{xx}&     \psi_{xy}\\     \psi_{xy}     &
  \psi_{yy}\end{pmatrix} 
\eeq
and consequently the invariants of this tensor are
$$\det(h_\omega)=2\Delta p, \;\;\;\tr(h_\omega) = 2\Delta \psi .$$
\end{Rem}

We  note that  metric, as  written in  (\ref{R1}), is  related to  the
velocity gradient tensor  (VGT) of a two-dimensional flow. The signature of the induced metric changes when the sign of the Laplacian of the pressure changes, and this is related to the flip between the elliptic/hyperbolic nature of (\ref{2DE}) when viewed as 
a MA equation for $\psi$.

The VGT is
not a  symmetric tensor and  its trace is  the divergence of  the flow
(which  equals  zero  in  the   incompressible  case).  The  trace  of
$h_\omega$  is  proportional  to  the   vorticity  of  the  flow.  The
determinant of  the VGT is the  same as $\det(h_\omega)$, and  this is
precisely  the balance  between the  rate-of-strain and  the enstrophy
(vorticity squared). It  is possible to construct  equations of motion
for the  invariants of the VGT  and for the invariants  of $h_\omega$: see Roulstone {\em et al.} (2014) for further details.

\section{Monge-Amp\`ere structures in six dimensions}

\subsection{Geometry of 3-forms}

In six dimensions ($n=3$), there is again a correspondence between
real/complex geometry and ``nondegenerate'' Monge-Amp\`ere
structures.

Lychagin {\em et al.} (1993)  associated   with
a Monge-Amp\`ere structure $(\Omega,\omega)$ an invariant symmetric form 
$$
g_\omega(X,Y) = \frac{\iota_X\omega\wedge\iota_Y\omega\wedge \Omega}{\vol}
$$
whose signature  distinguishes the different orbits  of the symplectic
group action.

In a seminal paper on the geometry of 3-forms,
  Hitchin (2001) defined  the  notion  of  nondegenerate  3-forms  on a
  6-dimensional  space and  constructed a  scalar invariant,  which we
  call the {\em Hitchin pfaffian}, which is non zero for such nondegenerate
  3-forms.   Hitchin also defined an  invariant tensor  $A_\omega$ on  the
  phase space satisfying 
$$
A_\omega^2 = \lambda(\omega)\,Id .
$$
Note that in the nondegenerate case, $\displaystyle K_\omega
=\frac{A_\omega}{\sqrt{|\lambda(\omega)|}}$   is  a   product  structure  if
$\lambda(\omega)>0$ and a complex structure if $\lambda(\omega)<0$.

Following Hitchin {\em op. cit.}, we observe
\begin{enumerate}
\item $\lambda(\omega)>0$  if and only if  $\omega$ is the sum  of two
  decomposable forms; i.e.
$$
\omega=\alpha_1\wedge\alpha_2\wedge\alpha_3+\beta_1\wedge\beta_2\wedge\beta_3
$$
\item $\lambda(\omega)<0$  if and only if  $\omega$ is the sum  of two
  decomposable complex forms; i.e. 
$$
\omega=
(\alpha_1+i\beta_1)\wedge(\alpha_2+i\beta_2)\wedge(\alpha_3+i\beta_3)+
(\alpha_1-i\beta_1)\wedge(\alpha_2-i\beta_2)\wedge(\alpha_3-i\beta_3)
$$
\end{enumerate}

Banos (2002) noted that on a $6$-dimensional symplectic space,
nondegenerate $3$-forms correspond  to nondegenerate Lychagin-Roubtsov
symmetric tensors. Moreover the signature  of $g_\omega$ is $(3,3)$ if
$\lambda(\omega)>0$ and $(6,0)$ or $(4,2)$ if $\lambda(\omega)<0$.

Banos {\em op. cit.} established the following simple relation between the Hitchin
and the Lychagin-Roubsov invariants 
$$
g_\omega(A_\omega\cdot,\cdot) = \Omega(\cdot,\cdot) .
$$
There is a correspondence between non-degenerate    MA   structures
$(\Omega,\omega)$ and a ``real Calabi-Yau structure'' if the Hitchin pfaffian
  $\lambda(\omega)$ is positive, and a ``complex Calabi-Yau structure'' if the Hitchin pfaffian is negative. 

For example, the Monge-Amp\`ere structure associated with the
"real" MA equation in three variables $(x_1,x_2,x_3)$
$$
\hess(\phi)=1\, ,
$$
is the pair
$$
\begin{cases}
\Omega=dx_1\wedge d\xi_1+ dx_2\wedge d\xi_2 + dx_3\wedge d\xi_3&\\
\omega=d\xi_1\wedge d\xi_2\wedge d\xi_3 -dx_1\wedge dx_2\wedge dx_3 &\\
\end{cases}
$$
and the underlying real analogue of K\"ahler structure on $T^*\mathbb{R}^3$ 
in the coordinates $(x_1,x_2,x_3,\xi_1,\xi_2,\xi_3)$ is
$$
\begin{cases}
g_\omega=\begin{pmatrix} 0&Id\\Id&0\end{pmatrix} ,&\\
K_\omega=\begin{pmatrix} Id&0\\0&-Id\end{pmatrix} .
\end{cases}
$$
The Monge-Amp\`ere structure associated with the special
lagrangian equation
$$
\Delta \phi - \hess(\phi)=0
$$
is the pair
$$
\begin{cases}
\Omega=dx_1\wedge d\xi_1 + dx_2\wedge d\xi_2 + dx_3\wedge d\xi_3 &\\
\omega=\im\big( (dx_1 +id\xi_1)\wedge (dx_2 +id\xi_2)\wedge (dx_3 +id\xi_3)\big)&\\
\end{cases}
$$
and the underlying K\"ahler structure is the canonical K\"ahler
structure on $T^*\mathbb{R}^3=\mathbb{C}^3$:
$$
\begin{cases}
g_\omega=\begin{pmatrix} Id&0\\0&Id\end{pmatrix} ,&\\
K_\omega=\begin{pmatrix} 0&-Id\\Id&0\end{pmatrix} .&\\
\end{cases}
$$

As in the $4$-dimensional case, there is a clear connection between the
problem of local equivalence of MA equation in three variables and the
integrability problem of generalized Calabi-Yau structures on
$\mathbb{R}^6$ (Banos 2002):

\begin{Prop}
The three following assertions are equivalent:
\begin{enumerate}[i)]
\item the MA equation $\Delta_\omega=0$ is locally equivalent to one of the three
equations
$$
\begin{cases}
\hess(\phi)=1 ,&\\
\Delta\phi-\hess(\phi)=0 ,&\\
\square \phi+\hess(\phi)=0 ,&\\
\end{cases}
$$
\item the underlying generalized Calabi-Yau structure on $T^*\mathbb{R}^3$ is integrable,
\item the form $\frac{\omega}{\sqrt[4]{|\lambda(\omega)|}}$ and
its dual form (in the sense of Hitchin) are closed. Moreover, the metric $g_{\omega}$ should be flat.
\end{enumerate}
\end{Prop}

The generalized Calabi-Yau structure discussed here is defined in Banos (2002). This notion is different from the generalized Calabi-Yau structures in sense of N. Hitchin and M. Gualtieri. These relations are clarified and discussed further in $\S$14.4
of Kosmann-Schwarzbach and Rubtsov (2010).

It is important to note that the geometry associated with a MA equation
$\Delta_\omega=0$  of real type ($\lambda(\omega)>0$) is
essentially real but it is very similar to the classic K\"ahler
geometry. In particular, when this geometry is integrable, there
exists a potential $\Phi$ and a coordinate system
$(x_i,u_i)_{i=1,2,3}$ on $T^*\mathbb{R}^6$ such that
$$
g_\omega=\underset{i,j}{\sum} \frac{\partial ^2 \Phi}{\partial
x_i\partial \xi_j} dx_i\cdot d\xi_j
$$
and
$$
\det\left(\frac{\partial^2 \Phi}{\partial x_j\partial
\xi_j}\right)=f(x)g(\xi).
$$

\subsection{Reduction principle in six dimensions}

The Marsden-Weinstein  reduction process is  a classical  tool in symplectic geometry to reduce spaces  with symmetries (see for example Marsden 1992). It  is  explained  in  Banos (2015)  how  the  theory  of  Monge-Amp\`ere operators  gives  a  geometric   formalism  in  which  the  symplectic reduction  process  of  Marsden-Weinstein  fits  naturally  to  reduce equations with symmetries, and we recall the basic ideas here.

Let us  consider a Hamiltonian action  $\lambda: G\times M\rightarrow
M$ of a $1$-dimensional  Lie group $G$ ($G=\mathbb{R}$, $\mathbb{R}^*$
or $S^1$) on  our symplectic manifold $(M=T^*\mathbb{R}^3,\Omega)$. We
denote by  $X$ the  infinitesimal generator of  this action  and $\mu:
M\rightarrow \mathbb{R}$ is the moment map of this action:
$$
\iota_X\Omega = -d\mu .
$$

For  a regular  value  $c$, we  know  from the Mardsen-Weinstein  reduction
theorem  that the  $4$-dimensional  reduced space  $M_c=\mu^{-1}(c)/G$
admits a natural symplectic  form $\Omega_c$ such that $\pi^*\Omega_c$
is   the    restriction   of   $\Omega$    on   $\mu^{-1}(c)$,   where
$\pi:\mu^{-1}(c)\rightarrow M_c$ is the natural projection.

Assume moreover that the action  $\lambda$ preserves also the $3$-form
$\omega$:
$$
\forall g\in G,\;\; \lambda_g^*(\omega) = \omega .
$$
Then there exists a $2$-form $\omega_c$ on $M_c$ such that
$$\pi^*(\omega_c)=\iota_X \omega .
$$

We then  obtain a  Monge-Ampère structure  $(\Omega_c,\omega_c)$ on  the
reduced space $M_c$,  and there is a  one-to-one correspondence between
regular  solutions  $L_c$ on  $M_c$  (that  is $\Omega_c|_{L_c}=0$  and
$\omega_c|_{L_c}=0$)    and   $G$-invariant    generalized   solutions
$L=\pi^{-1}(L_c)$ on $M$.

{\bf Example}
We consider as a trivial example the $3D$-Laplace equation 
$$
\phi_{x_1x_1}+\phi_{x_2x_2}+\phi_{x_3x_3}=0 .
$$
The corresponding Monge-Ampère structure is $(\Omega,\omega)$, which in local coordinates $x_i, \xi_i$ is
$$
\Omega=dx_1\wedge d\xi_1+dx_2\wedge d\xi_2+dx_3\wedge d\xi_3
$$
and 
$$
\omega=d\xi_1\wedge     dx_2\wedge    dx_3+     dx_1\wedge    d\xi_2\wedge
dx_3+dx_1\wedge dx_2\wedge d\xi_3 .
$$

Let $G=\mathbb{R}$  act on  $T^*\mathbb{R}^3$ by translation  on the
third coordinate:
$$
\lambda_\tau (x_1,x_2,x_3,\xi_1,\xi_2,\xi_3) = (x_1,x_2,x_3+\tau,\xi_1,\xi_2,\xi_3) .
$$
This    action   is    trivially   Hamiltonian    with   moment    map
$\mu(x,\xi )=-\xi_3$   and   infinitesimal  generator   $\displaystyle
X=\lambda_*\left(\frac{d}{d\tau}\right) =\frac{\partial}{\partial x_3}$.

For $c\in \mathbb{R}$, the reduced space $M_c$ is canonically 
$\mathbb{R}^4$   and    the   reduced   Monge-Ampère    structure   is
$(\Omega_c,\omega_c)$
 with
$$
\Omega_c=dx_1\wedge d\xi_1+dx_2\wedge d\xi_2
$$
and 
$$
\omega_c =  \iota_{\frac{\partial}{\partial x_3}} \omega  = d\xi_1\wedge
dx_2+dx_1 \wedge d\xi_2 .
$$
Then we obtain the Laplace equation in two dimensions as the reduced equation.

\section{The geometry of Burgers'-type vortices in three dimensions}

\subsection{3D-2D reduction in fluid dynamics: introduction}

It is well known in fluid dynamics that certain solutions of the incompressible Euler (and Navier--Stokes) equations in three dimensions can be mapped to solutions of the incompressible equations in two dimensions using the so-called Lundgren transformation. The class of solutions for which this mapping is possible is characterised by the linearity of the third component of velocity in the third spatial coordinate, with the other two components of the velocity being independent of the third spatial coordinate. The canonical Burgers' vortex is a member of this family of solutions, and in this section we show how the formulation described in Section~2.2 can be applied to illuminate the geometry of these solutions.

Joyce  (2005) used  the reduction  principle to  construct examples  of
$U(1)$- invariant special lagrangian manifolds in $\mathbb{C}^3$ after
reducing  the  6-dimensional  ``complex"  Calabi-Yau  structure.   The
$4$-dimensional geometry  can be seen  as a symplectic reduction  of a
$6$-dimensional ``real Calabi-Yau'' geometry, which we describe now. 

In this section we are  concerned with solutions of the incompressible
Euler equations  of Burgers'-type. That  is, we consider flows  of the
form 
\beq
\label{Bu}
u_i  = (-\gamma(t)x_1/2  - \psi_{x_2},  -\gamma(t)x_2/2 +  \psi_{x_1},
\gamma(t)x_3)^T, 
\eeq
which satisfy  incompressibility $u_{ii} =  0$, and $\psi$, a stream function, is independent of $x_3$. Here  $\gamma(t)$ is the strain, which is a
function of  time alone. We shall consider the  analogue of
(\ref{2DE}) for  this class of  flows in  three dimensions and,  to this end, it is
convenient to introduce a new stream function $\Psi(x_1, x_2, x_3, t)$
as follows 
\beq
\label{Psi}
\Psi = \psi(x_1, x_2, t) -\frac{3}{8}\gamma(t)^2 x_3^2 .
\eeq

\subsection{Extended real Calabi-Yau geometry}

We  now extend the  $4$-dimensional   geometry  associated   with $(\ref{2DE})$ to a $6$-dimensional geometry.

Using (\ref{Bu}) and (\ref{Psi}), (\ref{2DP}) becomes 
\begin{equation}
\Psi_{x_1x_1}\Psi_{x_2x_2}-\Psi_{x_1x_2}^2 +\Psi_{x_3x_3} = \dfrac{\Delta p}{2}
\label{E2DE}
\end{equation}
where the Laplace  operator on the right  hand side is
  in three dimensions, but note (from (\ref{Psi})) that $\Delta p$ is independent of $x_3$ (cf. Ohkitani and Gibbon 2000).

Denote by  $(\Omega,\varpi)$ the corresponding  Monge-Ampère structure
on $\mathbb{T}^*\mathbb{R}^3$, with $\Omega$ the canonical symplectic
form
$$
\Omega=dx_1\wedge d\xi_1 + dx_2 \wedge d\xi_2 + dx_3\wedge d \xi_3 ,
$$
where $\xi_i = \nabla_i\Psi$ and, following Roulstone  {\em et al.} (2009),  $\varpi$ is the effective $3$-form 
\begin{equation}
\varpi =  d\xi_1\wedge d\xi_2\wedge dx_3  -a dx_1\wedge dx_2\wedge  dx_3 +
dx_1\wedge 
dx_2\wedge d\xi_3,
\label{E2DEform}
\end{equation}
with $a(x_1,x_2) = \Delta p/2$.

The corresponding ``real Calabi-Yau''       structure       on       $T^*\mathbb{R}^6$       is
$$
\left(g_\varpi,K_\varpi,\Omega,
  \varpi+\hat{\varpi},\varpi-\hat{\varpi}\right)
$$
with
\begin{enumerate}[i)]
\item $\displaystyle g_\varpi = 2a dx_3\otimes dx_3 + dx_1\otimes d\xi_1 + dx_2\otimes d\xi_2 -dx_3\otimes
d\xi_3 \;\;\;\;\;\;\; \left(\varepsilon(g_\varpi)=(3,3)\right)$
\item $\displaystyle K_\varpi = \begin{pmatrix}
-1 & 0 & 0 & 0 & 0 & 0\\
0 & -1 & 0 & 0 & 0 & 0\\
0&0&1&0&0&0\\
0&0&0&1&0&0\\
0&0&0&0&1&0\\
0&0&2a&0&0&-1\\
\end{pmatrix} \;\;\;\;\;\;\; \mbox{and thus} \;\;\;\; K_\varpi^2 =1$ 
\item  $\displaystyle \varpi +\hat{\varpi}  = 2  d\xi_1\wedge d\xi_2  \wedge dx_3,
  \;\;\;\; \varpi -\hat{\varpi} = 2dx_1\wedge dx_2\wedge (d\xi_3-adx_3) ,$ 
\end{enumerate}
and the Hitchin pfaffian is $\lambda(\varpi)=1$.

\begin{Prop}
\begin{enumerate}
\item  This ``real  Calabi-Yau'' structure  is integrable:  locally there   exists a potential $F$ and coordinates $(x,\xi)$ in which
$$
g_\varpi   =  \sum_{j=1}^3\sum_{k=1}^3   \frac{\partial^2  F}{\partial
  x_j\partial \xi_k}dx_j\cdot d\xi_k
$$
and 
$$
\Omega = \sum_{j=1}^3\sum_{k=1}^3   \frac{\partial^2  F}{\partial
  x_j\partial \xi_k}dx_j\wedge d\xi_k
$$
\item The pseudo-metric $g_\varpi$ is Ricci-flat.
\item It is flat if and only if $\Delta p = \alpha_1 x_1+\alpha_2 x_2+ \beta$.
\end{enumerate}
\end{Prop}

\begin{proof}
Since     $d\varpi=d\hat{\varpi}=0$,     we     obtain     immediately
integrability.  Direct  computations  show   that  the  pseudo  metric
$g_\varpi$  has a  curvature  tensor  which is  null  if  and only  if
$a(x_1,x_2)=   \Delta   p/2$    is   "affine",   whereas
  Ricci-curvature always vanishes.

\end{proof}

\begin{Rem}
From Proposition 3, we deduce that equation \eqref{E2DE} is equivalent
to $\hess  \psi =1$ if and  only if $\Delta p  = \alpha_1 x_1+\alpha_2
x_2+ \beta$, where $\alpha_1, \alpha_2$ and $\beta$ are constants.
\end{Rem}

\subsection{The symplectic reduction}

Let $G=\mathbb{R}$  acting on  $T^*\mathbb{R}^3$ by translation  on the
third coordinate:
$$
\lambda_\tau (x_1,x_2,x_3,\xi_1,\xi_2,\xi_3) = (x_1,x_2,x_3+\tau,\xi_1,\xi_2,\xi_3)
$$
This    action   is    trivially   hamiltonian    with   moment    map
$\mu(x,\xi)=-\xi_3$   and   infinitesimal  generator   $\displaystyle
X=\lambda_*\left(\frac{d}{d\tau}\right) =\frac{\partial}{\partial x_3}$.

For $c\in  \mathbb{R}$, the reduced  space $M_c$ is the  quotient space
$$\begin{aligned}
M_c&=\mu^{-1}(c)/\mathbb{R}\\
&       =        \left\{(x_1,x_2,x_3,\xi_1,\xi_2,-c),       \;
  (x_1,x_2,x_3,\xi_1,\xi_2)\in\mathbb{R}^5\right\}/x_3\sim x_3+\tau \\
&\cong
\left\{(x_1,x_2,\xi_1,\xi_2,-c),\; (x_1,x_2,\xi_1,\xi_2)\in\mathbb{R}^4\right\}
\end{aligned}
$$

The vector $\displaystyle Y=K_\varpi X = \frac{\partial}{\partial x_3}
+2a\frac{\partial}{\partial    \xi_3}$     satisfies    $\Omega(X,Y)=2a$
with  $a(x_1,x_2)=\Delta   p  /   2$. Therefore, 
when $\Delta p\neq 0$, one can define a $4$-dimensional distribution
$D$ which is $\Omega$-orthogonal to $\mathbb{R}X\oplus \mathbb{R}Y$:
$$
T_m\mathbb{R}^6     =    D_m     \oplus_{\bot}\left(\mathbb{R}X_m\oplus
  \mathbb{R}Y_m\right) 
$$
with $m = (x,\xi)\in\mathbb{R}^6 = T^*\mathbb{R}^3$, and $\Omega$ and $\varpi$ decompose as follows:
$$
\begin{cases}
\Omega= \Omega_c -\dfrac{1}{2a} \iota_X\Omega \wedge \iota_Y\Omega&\\
&\\
\varpi    =    \varpi_1\wedge     \iota_X\Omega    +    \varpi_2\wedge
\iota_Y\Omega .&\\
\end{cases}
$$
Here, $\Omega_c, \varpi_1$ and $\varpi_2$ in $\Omega^2(D)$ are defined by 
$$
\begin{cases}
\Omega_c=dx_1\wedge d\xi_1 + dx_2\wedge d\xi_2&\\
&\\
\varpi_1 = -\dfrac{1}{2a} (d\xi_1\wedge d\xi_2 - adx_1\wedge dx_2),&\\
&\\
\varpi_2 = \dfrac{1}{2a}(d\xi_1\wedge d\xi_2 + a dx_1\wedge dx_2).&\\
\end{cases}
$$

Using  the natural  isomorphism  $T_{[m]}M_c \cong  D_m$,  we get  a
triple  $(\Omega_c,\varpi_1,\varpi_2)$  on $M_c$  which  is,  up to  a
re-normalization,  the  hypersymplectic   triple  on  $T^*\mathbb{R}^4$ 
associated with the $2D$-Euler equation (cf. $\S$3.2).

\section{Reduction of incompressible 3D-Euler equations}

\subsection{Introduction: Monge-Amp\`ere geometry of a non-Monge-Amp\`ere equation}

We now look at more  general solutions to the incompressible equations
in three dimensions. The Burgers' vortex is a special solution in that we may define a stream function (\ref{Psi}) and use the theory described in $\S$5.2 
and Banos (2002). However, for general incompressible flows in three dimensions a single scalar stream function is not available and we have to work with a 
vector potential. Therefore we should  stress that there is no underlying MA equation. We now proceed to show how a geometry, very similar to the one described earlier, can be recovered for three-dimensional flows with symmetry. Our results extend those of Roulstone {\em et al.} (2009b).

We introduce the following 3-forms $\omega$ and $\theta$ on the phase
space $T^*\mathbb{R}^3$, with local coordinates $(x_i, u_i)$ (so here we treat the $u_i$ as a velocity, and these coordinates replace the $\xi_i$; more rigorously, we have identified $T\mathbb{R}^3$ and $T^*\mathbb{R}^3$ using the euclidean metric):
\begin{equation}
\label{3Domega}
\omega = a\,dx_1\wedge  dx_2\wedge dx_3 - du_1 \wedge
du_2 \wedge  dx_3 - du_1 \wedge  dx_2 \wedge du_3 -  dx_1 \wedge du_2
\wedge du_3 
\end{equation}
with  $$a=\frac{\Delta p}{2}$$ 
and 
\begin{equation}
\label{3Dtheta}
\theta = du_1 \wedge dx_2 \wedge dx_3 + dx_1 \wedge du_2 \wedge dx_3 +
dx_1 \wedge dx_2 \wedge du_3 .
\end{equation}

Considering        $u=(u_1,u_2,u_3)$        as       a        $1$-form
$u=u_1dx_1+u_2dx_2+u_3dx_3$ on  $\mathbb{R}^3$, it is straightforward
to check  that $u$ is a solution of $-\Delta p = u_{ij}u_{ji},\;\;u_{ii}=0$ if
and  only  if   the  graph $L_u  =  \left\{(x,u)\right\}\subset
T^*\mathbb{R}^3$  is  bilagrangian with  respect  to $\omega$  and
$\theta$; that is 
\begin{equation}
\label{3DEF}
\omega|_{L_u} = 0 \;\;\;\;\; \text{ and } \;\;\;\;\; \theta|_{L_u}=0 .
\end{equation}
We observe that  a solution here  is not
  necessarily  lagrangian with  respect to  the underlying  symplectic
  form $\Omega$, which is compatible with the fact that our equation is no longer of Monge-Amp\`ere form.

\subsection{Hitchin tensors and Lychagin-Roubtsov metrics}

Note that $$\omega\wedge \theta =  3\, \text{vol} \;\;\; \text{ with }
\;\;\;   \text{vol}=  dx_1\wedge   dx_2\wedge   dx_3\wedge  du_1\wedge
du_2\wedge du_3 .$$ 
We can
therefore define the following analogues of Hitchin tensors $K_\omega$ and $K_\theta$:
$$
<\alpha,K_\omega(X)>=              \frac{\iota_X(\omega)\wedge\omega\wedge
  \alpha}{\text{vol}}\;\;\; ,\;\;\; <\alpha,K_\theta(X)> =
\frac{\iota_X(\theta)\wedge\theta\wedge 
  \alpha}{\text{vol}} 
$$
for any $1$-form $\alpha$ and any vector field $X$.

In the $3\times 3$ decomposition $(x,u)$, these tensors are
\begin{equation}
\label{K}
K_\omega=   2\begin{pmatrix}   0  &   -\text{Id}\\   a\,  \text{Id}   &
  0\end{pmatrix}\;,\;\;\; K_\theta= 2\begin{pmatrix} 0 & 0\\ \text{Id}
  & 0\end{pmatrix} .
\end{equation}
We have the relations
\begin{equation}
\label{K_relations}
\begin{aligned}
K_\omega^2 &= -4a\,\text{Id}\\
K_\theta^2 &= 0\\
K_\omega K_\theta & +K_\theta K_\omega = -4\text{Id}\\
[K_\omega & ,K_\theta] = 4\begin{pmatrix} -\text{Id} & 0 \\ 0 &
  \text{Id}\end{pmatrix}\\
\end{aligned}
\end{equation}
Moreover, using the symplectic form $\Omega=dx_1\wedge du_1+dx_2\wedge
du_2  + dx_3\wedge du_3$,  one can  also define the analogues of Lychagin-Roubtsov
metrics  $g_\omega  =  \Omega(K_\omega\cdot,\cdot)$  and  $g_\theta  =
\Omega(K_\theta\cdot,\cdot)$:
\begin{equation}
\label{G}
g_\omega=\begin{pmatrix}     2a\,\text{Id}    &     0     \\    0     &
  2\text{Id}\end{pmatrix}   \;\;   ,   \;\;   g_\theta=\begin{pmatrix}
  2\text{Id} & 0 \\ 0 & 0 \\ \end{pmatrix} .
\end{equation}

\subsection{The symplectic reduction}

Let us consider  an hamiltonian action $\lambda$  of a $1$-dimensional
Lie group $G$ on our $6$-dimensional phase space with moment map $\mu$.
Assume moreover that
\begin{enumerate}
\item the action $\lambda$ preserves also $\omega$ and $\theta$:
$$
\forall g\in G,\lambda_g^* (\omega) = \omega\; \;\; \text{ and } \;\;
\lambda_g^* (\theta) = \theta
$$
\item the infinitesimal  generator $X$ satisfies $g_\omega(X,X)\neq 0$
  and $g_\theta(X,X)\neq 0$
\end{enumerate}
NB. We impose that the action preserves the
  symplectic form $\Omega$  to obtain a reduced space, but  as we already
  mentioned,  this symplectic  form  is additional data  in  this   context.

As  we  have  seen,  there  exists  a  pair  $(\omega_c,\theta_c)$  of
$2$-forms on $M_c$ 
such that 
\begin{equation}
\label{reduction}
\pi^*(\omega_c)=\iota_X \omega\;\;,\;\; \pi^*(\theta_c)=\iota_X \theta
\end{equation}

\begin{Rem} if $L_c$ is  a $2$-dimensional bilagrangian submanifold of
  $M_c$  (that is  $\omega_c|_{L_c}=0$  and $\theta_c|_{L_c}=0$)  then
  $L=\pi^{-1}(L_c)$  is ``bilagrangian'' with  respect to  $\omega$ and
  $\theta$.

Hence, there is a  correspondence between bilagrangian surfaces in the
reduced space $M_c$ and $G$-invariant solutions of \eqref{3DEF}.
\end{Rem}

\begin{Rem}
One can check that the largest group preserving $\Omega$, $\omega$ and
$\theta$ is $SO(3,\mathbb{R})$ acting on $M$ by 
$$
A\cdot (x,u) = (A\cdot x,A^{-1} \cdot u) .
$$
\end{Rem}

\subsection{Action by translation} 

Let $\mathbb{R}$ acting on $T^*\mathbb{R}^3$ by translation on $(x_3,u_3)$:
$$
\lambda_a(x_1,x_2,x_3,u_1,u_2,u_3) = (x_1,x_2,x_3+\tau,u_1,u_2,u_3+\gamma
\tau)\;,\;\;\;\; \gamma\in\mathbb{R}
$$
The infinitesimal generator is $$X=\lambda_*\left(\frac{d}{d\tau}\right)
= \frac{\partial}{\partial x_3} + 
\gamma \frac{\partial}{\partial u_3}$$ and the moment map is 
$$
\mu(x,u)=\gamma\,x_3-u_3 .
$$
We observe that $\mu = \mbox{constant}$ yields the linearity of $u_3$ in $x_3$, as defined in (\ref{Bu}).

For   $c\in\mathbb{R}$,   the  reduced   space   $M_c$  is   trivially
$T^*\mathbb{R}^2$:
$$
\begin{aligned}
M_c & = \left\{(x_1,x_2,x_3,u_1,u_2,\gamma x_3 -c\right\}/ x_3\sim
x_3+t\\
& =  \left\{(x_1,x_2,0,u_1,u_2,-c)\right\}
\end{aligned}
$$
and the pair $(\omega_c,\theta_c)=\left(\iota_X\omega,\iota_X\theta\right)$ is 
$$
\begin{cases}
\omega_c =  a\,dx_1\wedge dx_2 -  du_1\wedge du_2 -  \gamma du_1\wedge
dx_2 - \gamma dx_1\wedge du_2&\\
\theta_c = du_1\wedge dx_2+dx_1\wedge du_2+\gamma dx_1\wedge dx_2&\\
\end{cases}
$$

Considering the following change of variables 
$$
\begin{aligned}
X_1 = x_1\;, & \;\;\;\;\;\;\;\;U_1 = \frac{\gamma}{2} x_2 + u_2,\\
X_2 = x_2 \;,& \;\;\;\;\;\;\;\; U_2 = -\frac{\gamma}{2} x_1 - u_1,\\
\end{aligned}
$$
we obtain
\begin{equation}
\label{tc}
\theta_c=dX_1\wedge dU_1+dX_2\wedge dU_2
\end{equation}
and
\begin{equation}
\label{oc}
\omega_c = \omega_{0} -\frac{\gamma}{2}\theta_c
\end{equation}
with 
\begin{equation}
\label{o0}
\omega_{0} = \left(a+\frac{3}{4}\gamma^2\right)dX_1\wedge dX_2 - dU_1\wedge dU_2 .
\end{equation}

Hence we summarize our result as follows:
\begin{Prop}
If $\psi(x_1,x_2,t)$ is solution of the Monge-Ampère equation in two dimensions
\begin{equation}
\label{3DER1}
\psi_{x_1x_1}\psi_{x_2x_2}-\psi_{x_1x_2}^2     =      \frac{\Delta     p(x_1,x_2,t)}{2}     +
\frac{3}{4}\gamma^2(t) .
\end{equation}
then the velocity $u(x_1,x_2,x_3,t)$ defined by 
$$
\begin{aligned}
u_1& = -\frac{\gamma(t)}{2} x_1 -\psi_{x_2},\\
u_2&=  -\frac{\gamma(t)}{2} x_2 +\psi_{x_1},\\
u_3& = \gamma(t) x_3 -c(t),\\
\end{aligned}
$$
is solution of $-\Delta p = u_{ij}u_{ji}$ and $u_{ii} = 0$ in three dimensions.
\end{Prop}

\begin{proof}
Time $t$  is a parameter  here, hence $c=c(t)$  and $\gamma=\gamma(t)$
are constant.
If $\psi(x_1,x_2,t)$ is solution of \eqref{3DER1} then
$\omega_0|_{\psi}=\theta_c|_{L_\psi}=0$,      with       
$$L_\psi      =
\{(X_1,X_2,\psi_{X_1},\psi_{X_2})\}.$$ From \eqref{oc}, we deduce that
$\omega_c|_{L\psi}   =    \theta_c|_{L_\psi}   =0$,    and   therefore
$\omega|_L=\theta|_L = 0$ with 
$$L=\pi^{-1}(L_c)=\left\{(x_1,x_2,x_3,-\frac{\gamma}{2}x_1    -   \psi_{x_2},
  -\frac{\gamma}{2}x_1+\psi_{x_1}, \gamma x_3 - c)\right\} .
$$
\end{proof}

\section{Summary}

The Poisson equation for the pressure in the incompressible Euler and Navier--Stokes equations can be written in the form
\beq
\label{peqn}
\Delta p = \frac{\zeta^2}{2} - \mbox{Tr}S^2,
\eeq
where $\zeta$ is the vorticity and $S$ is the rate-of-strain matrix. Equation (\ref{peqn}) illuminates the relationship between the sign of the Laplacian of the pressure and the balance between vorticity and strain in the flow. Gibbon (2008) has conjectured that much might be learned from studying the geometric properties of this equation, which {\em locally} holds the key to the formation of vortical structures --- ubiquitous features of both turbulent flows and large-scale atmosphere/ocean dynamics.

In this paper we have studied (\ref{peqn}) through Monge-Amp\`ere structures and their associated geometries. Burgers' vortex, which has been described as the {\em sinew of turbulence}, is shown to arise naturally from a symmetry reduction of a Monge-Amp\`ere structure (Section 5). This Monge-Amp\`ere structure is itself derived from a symmetry reduction of the equations for incompressible flow in three dimensions ($\S$6.3), and mapping to a solution of 2d incompressible flow has been presented (Proposition 5). This is surely associated with the well-known Lundgren transformation. 

In contrast to earlier works by the same authors, we have shown explicity how the geometries associated with both signs of the Laplacian of the pressure can be described by almost-complex and almost-product structures. This points to a possible role for generalised geometry (Banos 2007), and this is a topic for future research.

\section{Acknowledgements}

The paper was started when I.R and V.N.R. were participants of the Isaac Newton Institute for Mathematical Science Programme "Mathematics for the Fluid Earth" (21st October to 20th December 2013).
They are grateful to the organisers and participants of the Programme for the invitation and  useful discussions, and to the administration and the staff of the Newton Institute for hospitality, stimulating atmosphere and support during their stay in Cambridge. 

The final results of the paper were obtained while all three authors were participants of the Clay Mathematics Institute (CMI) meeting "Geometry and Fluids"  at Oxford University (7th-11th April 2014). They are thankful to the organisers and the CMI for the invitation and possibility to present their result during the Meeting. The discussions with M. Gualtieri, N.J. Hitchin, J. McOrist and M. Wolf were very inspiring and useful.
The work of V.N.R. were (partly) supported by the grants RFBR-12-01-00525-a and  RFBR 15-01-05990.

\section{References}

\noindent
Banos, B  2002 Nondegenerate  Monge-Ampère structures in  dimension 6,
{\em Letters in Mathematical Physics}, {\bf 62}, 1-15 

\noindent
Banos, B. 2007
Monge-Ampere equations and generalized complex geometry. The two-dimensional case. {\em Journal of Geometry and Physics}, {\bf  57}, 841-853

\noindent 
Banos, B 2015 Symplectic Reduction  of Monge-Ampère equations, {\em in   preparation} 

\noindent
Delahaies, S and Roulstone, I (2010) Hyper-K\"ahler geometry and semi-geostrophic theory. {\em Proc. R. Soc. Lond.} A, {\bf 466}, 195 - 211.

\noindent
Gibbon, J.D. 2008 The three-dimensional Euler equations: Where do we stand? {\em Physica D}, {\bf 237}, 1894-1904. DOI: 10.1016/j.physd.2007.10.014

\noindent
Hitchin N. J. 1990 Hypersymplectic quotients, {\em Acta
  Acad. Sci. Tauriensis}, {\bf 124}, 169-80

\noindent
Hitchin N. J. 2001 The geometry of three-forms in
    six and seven dimensions, {\em J. Differential Geometry}, {\bf 55}, 547-76

\noindent
Joyce,    J    2005    $U(1)$-invariant    special    lagrangian
  $3$-folds: I. Nonsingular solutions, {\em Adv. Math.}, {\bf 192}, 35-71
  
\noindent
Kosmann-Schvarzbach, Y. Rubtsov, V. 2010 Compatible structures on Lie algebroids and Monge-Ampère operators.
\emph{ Acta Appl. Math.}, {\bf 109}, no. 1, 101-135 

\noindent
Kossowski, M. 1992 Prescribing invariants of lagrangian surfaces. {\em Topology}, {\bf 31}, 337-47

\noindent
Kushner, A., Lychagin, V. , Rubtsov, V. 2007 \emph{Contact geometry and non-linear differential equations.} Encyclopedia of Mathematics and its Applications, 101. Cambridge University Press, Cambridge, 2007. xxii+496 pp.

\noindent
Larchev\^eque, M. 1993 Pressure field, vorticity field, and coherent structures 
in two-dimensional incompressible turbulent flows.
{\em Theor. Comp. Fluid Dynamics}, {\bf 5}, 215-22.

\noindent
Lundgren, T. 1982 Strained spiral vortex model for turbulent
fine structure. \textit{Phys. Fluids}, \textbf{25}, 2193-2203.

\noindent
Lychagin, V. V. 1979 Nonlinear differential equations and contact geometry. 
{\em Usp\`ekhi Mat. Nauk.}, {\bf 34}, 137-65.

\noindent
Lychagin, V.V. \& Rubtsov, V 1983a On Sophus Lie
    Theorems for Monge-Amp\`ere equations, {\em Doklady Bielorrussian
  Academy of Science}, {\bf 27}, no.5, 396-398 (in Russian)

\noindent
Lychagin, V.V. \& Rubtsov, V. 1983b Local
    classifications of Monge-Amp\`ere equations, {\em Soviet. Math. Dokl},
  {\bf 272}, no.1, 34-38.

\noindent
Lychagin, V. V., Rubtsov, V. N. \& Chekalov, I. V. 1993 A classification 
of Monge--Amp\`ere equations. {\em Ann. Sci. Ec. Norm. Sup.}, {\bf 26}, 281-308.

\noindent
McIntyre, M. E. \& Roulstone, I. 2002 Are there
higher-accuracy analogues of semi-geostrophic theory? In {\em
Large-scale atmosphere---ocean dynamics, Vol. II: Geometric
methods and models}.  J. Norbury and I. Roulstone (eds.); Cambridge: University Press.

\noindent
Marsden J.E. 1992  \emph{Lectures on  Mechanics}, Cambridge
  University Press.

\noindent
Ohkitani, K and Gibbon, J.D. 2000 Numerical study of singularity formation in a class of Euler and Navier-Stokes flows, {\em Physics of Fluids}, {\bf 12}, 3181-94.

\noindent
Roubtsov, V. N. \& Roulstone, I. 2001 Holomorphic structures in hydrodynamical 
models of nearly geostrophic flow. {\em Proc. R. Soc. Lond.}, A {\bf 457}, 1519-1531.

\noindent
Roulstone, I, Banos, B, Gibbon, JD and Roubtsov, V 2009a K\"ahler geometry and Burgers' vortices, {\em Proc. Ukrainian Nat'l Acad. Math.}, {\bf 16}, 303 - 321.

\noindent
Roulstone, I., Banos, B., Gibbon, J.D. \& Rubtsov, V.N. 2009b A geometric interpretation of coherent structures in Navier-Stokes flows. {\em Proc. R. Soc. Lond. } A {\bf 465}, 2015-2021 

\noindent
Roulstone, I., Clough, S.A., \& White, A.A. 2014 Geometric invariants of the horizontal velocity gradient tensor and their dynamics in shallow water flow. {\em Quart. J. R. Meteorol. Soc.}, {\bf 140}, 2527-34

\noindent
Rubtsov, V. N. \& Roulstone, I. 1997 Examples of
quaternionic and K\"ahler structures in Hamiltonian models of
nearly geostrophic flow. {\em J. Phys. A}, {\bf 30}, L63-L68.

\noindent
Weiss J. 1991 The dynamics of enstrophy transfer in two-dimensional hydrodynamics. {\em Physica D}. {\bf 48}, 273-94.

\end{document}